\documentclass[11pt]{article}

\usepackage{fullpage}
\usepackage{amssymb,amsmath}
\usepackage{epsfig}
\usepackage{palatino}
\usepackage{graphicx}
\usepackage{textcomp}
\usepackage{amsthm}
\usepackage{mdframed}
\usepackage{color, caption}
\usepackage{tocloft}
\usepackage[colorlinks=true,linkcolor=blue,linktoc=all]{hyperref}
\newtheorem{theorem}{Theorem}
\newtheorem{proposition}[theorem]{Proposition}

\newtheorem{definition}[theorem]{Definition}
\newtheorem{algorithm}[theorem]{Algorithm}

\newtheorem{lemma}[theorem]{Lemma}

\newtheorem{corollary}[theorem]{Corollary}

\newtheorem*{theorem*}{Theorem}
\newtheorem*{lemma*}{Lemma}
\newcommand{\Det}{\text{Det}}
\newcommand{\sgn}{\text{sgn}}

\newcommand{\nc}{\newcommand}
\nc{\rnc}{\renewcommand}

\def\ba#1\ea{\begin{align}#1\end{align}}
\def\bas#1\eas{\begin{align*}#1\end{align*}}
\def\bpm#1\epm{\begin{pmatrix}#1\end{pmatrix}}
\nc{\nn}{\nonumber}
\nc{\eq}[1]{(\ref{eq:#1})}
\nc{\eqs}[2]{(\ref{eq:#1}) and (\ref{eq:#2})}

\def\begsub#1#2\endsub{\begin{subequations}\label{eq:#1}\begin{align}#2\end{align}\end{subequations}}
\nc\qand{\qquad\text{and}\qquad}
\nc\mnb[1]{\medskip\noindent{\bf #1}}

\nc\benum{\begin{enumerate}}
\nc\eenum{\end{enumerate}}
\newcommand{\bea}{\begin{eqnarray}}
\newcommand{\eea}{\end{eqnarray}}
\nc\bit{\begin{itemize}}
\nc\eit{\end{itemize}}
\nc{\ot}{\otimes}
\rnc{\L}{\left}
\nc{\R}{\right}

\def\bbC{\mathbb{C}}

\def\bbF{\mathbb{F}}

\def\bbR{\mathbb{R}}

\def\R{\mathbb{R}}

\def\C{\mathbb{C}}

\def\poly{{\rm poly}}
\def\log{{\rm log}}

\def\Per{\rm Per}

\newcommand{\be}{\begin{eqnarray}}
\newcommand{\ee}{\end{eqnarray}}

\newcommand\floor[1]{{\lfloor #1 \rfloor}}
\newcommand\ceil[1]{{\lceil #1 \rceil}}

\def\NC{{\sf{NC}}}
\def\PL{{\sf{PL}}}
\def\PP{{\sf{PP}}}
\newcommand\Lspace{{\sf{L}}}
\newcommand\BPL{{\sf{BPL}}}
\newcommand\BQL{{\sf{BQL}}}
\newcommand\RL{{\sf{RL}}}
\newcommand\NL{{\sf{NL}}}
\newcommand\DET{{\sf{DET}}}

\newcommand{\ignore}[1]{}

\newcommand{\eps}{\varepsilon}
\renewcommand{\epsilon}{\varepsilon}

\nc{\hin}{h_{\text{in}}}
\nc{\pin}{\partial_{\text{in}}}
\nc{\pell}{\partial_{\ell}}

\newcommand{\nocontentsline}[3]{}
\newcommand{\tocless}[2]{\bgroup\let\addcontentsline=\nocontentsline#1{#2}\egroup}
\makeatletter
\newcommand{\cftsectionprecistoc}[1]{\addtocontents{toc}{%
  {\leftskip \cftsecindent\relax
   \advance\leftskip \cftsecnumwidth\relax
   \rightskip \@tocrmarg\relax
   \textit{#1}\protect\par}}}
\makeatother

\begin{document}

\title{
		\huge Approximating the Determinant of Well-Conditioned Matrices by Shallow Circuits\\
}
\author{Enric Boix{-}Adser\`a \thanks{MIT EECS, eboix@mit.edu}, Lior Eldar \thanks{eldar.lior@gmail.com}, Saeed Mehraban
\thanks{IQIM Caltech, mehraban@caltech.edu}}

\date{\today}

\maketitle

\abstract{
The determinant can be computed by classical circuits of depth $O(\log^2 n)$, and therefore it can also be computed in classical space $O(\log^2 n)$. Recent progress by \cite{Ta13} implies a method to approximate the determinant of Hermitian matrices with condition number $\kappa$ in quantum space $O(\log n + \log \kappa)$. However, it is not known how to perform the task in less than $O(\log^2 n)$ space using classical resources only. In this work, we show that the condition number of a matrix implies an upper bound on the depth complexity (and therefore also on the space complexity) for this task:
the determinant of Hermitian matrices with condition number $\kappa$ can be approximated to inverse polynomial relative error with classical circuits of depth $\tilde O(\log n \cdot \log \kappa)$, and in particular one can approximate the determinant for sufficiently well-conditioned matrices in depth $\tilde{O}(\log n)$.
Our algorithm combines Barvinok's recent complex-analytic approach for approximating combinatorial counting problems \cite{Bar16} 
with the depth-reduction theorem for low-degree arithmetic circuits \cite{valiant1983fast}.

}

\section{Introduction}

\subsection{Background}

\subsubsection{Exact computation of the determinant}
Computing the determinant of a matrix is arguably one of the most basic operations in numerical linear algebra and is ubiquitous in 
many areas of science and engineering.
As such, it has been intensively researched over the years with landmark results that have reduced its complexity,
improved its numerical stability and increased its parallelism (see for example a survey at \cite{KALTOFEN04}).
Today, we know that given an $n \times n$ matrix we can compute the determinant in time $O(n^{\omega})$ where $\omega < 2.373$
is such that $O(n^\omega)$ is the cost of matrix multiplication \cite{williams2012multiplying,le2014powers}.

One can also try to parallelize the computation of the determinant using many processors. We know that the computation of the determinant is in $\NC^2$ \cite{B84} -- that is, it can be computed
by circuits of depth $O(\log^2 n)$. While this implies that the determinant is parallelizable, it is a major open question whether or not the determinant can be parallelized
even further - namely, for instance, whether the determinant lies in $\NC^1$, meaning that it can be computed by circuits of depth $O(\log(n))$. Letting $\DET$ denote the class of problems that are $\NC^1$-reducible to computing the determinant of an integer-valued matrix, we know that
\begin{equation}
\NC^1 \subseteq \Lspace \subseteq \RL \subseteq \NL \subseteq \DET \subseteq \NC^2  \tag*{\cite{csanky1975fast, cook1985taxonomy}}.
\end{equation}
In particular, a $O(\log(n))$-depth circuit for the determinant would imply $\NC^1 = \Lspace = \RL = \NL$, which would be a major breakthrough in our understanding of space-bounded computation. Furthermore, since the class $\DET$ captures many fundamental linear-algebraic tasks such as matrix powering, matrix inversion, and solving systems of linear equations, a faster algorithm for exact computation of the determinant would have far-reaching applications \cite{cook1985taxonomy}.

\subsubsection{Approximate Computation of the Determinant} \label{subsubsec:approxdetcomputation}
In this paper, instead of exact computation, we consider the problem of {\em approximately} computing the determinant up to a $(1 + 1/\poly(n))$ multiplicative factor. The purpose of this section is to provide an overview of known bounds on the complexity of this task prior to this paper.

 The approximation problem is trivially reducible to $\DET$ and hence is contained in $\NC^2$. Interestingly, it turns out that merely computing the sign of the determinant is complete for the complexity class probabilistic logspace ($\PL$) \cite{AO96}. $\PL$ is the logspace analog of the class $\PP$, contains $\NL$ and is contained in $\NC^2$. As a result, similarly to exact computation, providing an $\Lspace$ algorithm for determinant approximation would imply surprising statements such as $\Lspace = \NL = \PL$. Hence, we would like to ask a more fine-grained question: Can the determinant be approximated using small space or depth on special inputs? 

The answer turns out to concretely depend on the degree to which a matrix is singular. In more precise terms, it depends on condition number of the input matrix, which is the largest-to-smallest singular value ratio of the matrix. 
Computing the sign of the determinant is complete for $\PL$ if the matrix is allowed to be arbitrary, in which case the condition number can be exponentially large. However, a result of Ta-Shma \cite{Ta13} shows that inverting a matrix with polynomially large condition number is contained the complexity class $\BQL$. As we describe in Appendix \ref{app:psddetreductiontosolving}, the techniques of Ta-Shma imply that $\BQL$ can approximate the determinant for matrices with polynomially large condition numbers. One may conjecture that the determinant approximation problem for polynomially-conditioned matrices is complete for $\BQL$. An item of evidence in favor of this conjecture is a result of Fefferman and Lin \cite{fefferman}, who show that inverting polynomially-conditioned matrices is complete for the class $\BQL$.\footnote{And while it is known that approximate inversion is $\NC^1$-reducible to determinant approximation through Cramer's rule, this reduction does not immediately work within the class of well-conditioned matrices considered by Fefferman and Lin.}

If a polynomially-conditioned matrix has certain special structure, then the absolute value of its determinant may be approximable by a weaker classical machine. This follows from an $\NC^1$-reduction that we give in Appendix \ref{app:psddetreductiontosolving} from approximating the absolute value of the determinant to approximate matrix inversion. Implications of our reduction include: 
\begin{itemize}\item if $A$ is a $\poly(n)$-conditioned symmetric diagonally dominant (SDD) matrix, then $\Det(A) = |\Det(A)|$ can be approximated in near-$\Lspace$, because there is a nearly-logspace SDD solver for $A$ \cite{murtagh2017derandomization}. \item if $A$ is a $\poly(n)$-conditioned Hermitian stochastic matrix, then $|\Det(A)|$ can be approximated in $\BPL$. This follows by combining the $\BPL$-algorithm of Doron et al. \cite{doron2017approximating} for approximating powers of $A$ with a parallelized version of the gradient descent algorithm for solving $Ax = b$.  
\item if $A$ is a $\kappa$-conditioned matrix, then $|\Det(A)|$ can be approximated in $\tilde{O}(\log(n)\cdot \log (\kappa))$ depth, because equations of the form $Ax = b$ can be solved in $\tilde{O}(\log(n) \cdot \log (\kappa))$ depth using parallelized gradient descent.
\end{itemize}
Details on the reduction from approximate calculation of $|\Det(A)|$ to approximate inversion of $A$ are given in Appendix \ref{app:psddetreductiontosolving}, but we note that an important barrier to this technique is the computation of the sign of the determinant, even for a Hermitian matrix. It is a priori unclear how to compute the signed determinant $\Det(A)$ with a classical algorithm running in space less than $\log^2(n)$, even if it has condition number $\kappa = O(1)$.

The main contribution of this paper is to surmount this barrier in computing the sign. For example, we show that the signed determinant of $\poly\log(n)$-conditioned matrices can be computed in nearly-logarithmic depth or nearly-logarithmic space if either the matrices are Hermitian or Hurwitz stable (Hurwitz stable matrices are defined as those that matrices that have eigenvalues with negative real parts).

\subsection{Main Results}

In this work we improve on prior art by proposing an algorithm for computing the determinant of well-conditioned $n\times n$ 
Hermitian or Hurwitz stable
matrices with condition number $\kappa$ that runs in depth
$$
\tilde{O}(\log(n) \cdot \log(\kappa))
$$
A matrix is Hurwitz stable if the real parts of its eigenvalues are negative. 
\begin{theorem}(sketch)
Let $A$ be an $n\times n$  Hermitian or Hurwitz stable matrix
with condition number $\kappa$.
There exists a Boolean circuit for approximating $\Det(A)$ to multiplicative error $1 + 1/\poly(n)$
that has $\tilde O(\log(n) \cdot \log (\kappa))$ depth. 
The circuits for this algorithm can be computed by a Turing machine using $\log(n)$-space.
\end{theorem}

\noindent
A direct corollary is the following:
\begin{corollary}(sketch)
Let $A$ be an $n \times n$ Hermitian or Hurwitz stable matrix with condition number $\kappa$. There exists an algorithm for approximating $\Det(A)$ to multiplicative error $1 + 1/\poly(n)$ that uses $\tilde{O}(\log(n) \cdot \log (\kappa))$ space.
\end{corollary}

\subsection{Proof Overview}

Our algorithm is inspired by the Taylor-series approach to computing a multiplicative approximation of the permanent, pioneered by Barvinok \cite{Bar16}.
In this approach, the permanent is reduced from a degree-$n$ polynomial in its entries to a univariate polynomial
as follows:
$$
g_A(z) = \Per(  (1 - z) \cdot J + z \cdot A),
$$
where $J$ is the all ones matrix. 
The polynomial $g_A(z)$ admits a Taylor series decomposition which converges to the true value of the function,
and in particular at $z=1$ -- namely $\Per(A) = g_A(1)$ which is our target -- whenever all the roots of $g_A(z)$ are
bounded away from the unit disk.

In order to compute a multiplicative approximation of $\Per(A)$, Barvinok considers $f_A(z) = \log(g_A(z))$
and computes an {\it additive} approximation of $f_A(z)$ at $z=1$ for any matrix $A$ for which the roots of $g_A(z)$
are bounded away from the unit disk.
For this algorithm, the Taylor series of $f(z)$ needs to contain only $O(\log(n))$ terms in order to provide
a $1/\poly(n)$ additive approximation error for $\log(\Per(A))$.
The algorithm then computes all low-order derivatives of $g(z)$ in time $2^{O(\log^2(n))}$,
uses them to compute all low-order derivatives of $f(z)$ at $z=0$, and then computes $f(1)$ as a Taylor-series approximation and finally returns $e^{f(1)}$.

Barvinok's approach was used in recent years to show a quasi-polynomial-time algorithm for computing the permanent of special cases
of matrices \cite{Barvinok2013, Barvinok2016, Bar16}, and to provide better approximations of partition functions \cite{patel2017deterministic,liu2019ising,liu2019fisher,mann2019approximation,harrow2019classical}.
In particular, authors 2 and 3 of this paper showed how to approximate the permanent of most Gaussian matrices
by considering a random ensemble of such matrices with a vanishing, non-zero mean \cite{EM18}.

The determinant is an $n$-variate polynomial that is very similar to the permanent, at least as a syntactic polynomial,
with permutation signs multiplying each monomial.
Hence it is natural to consider the determinant as a candidate for applying the Taylor-series approach.
However, a polynomial-time algorithm is already known for the determinant and this approach will not do any better. Our goal, instead,
is to focus on the {\it depth} complexity of the algorithm, which we then use to conclude a space-efficient log-space
algorithm for approximating the determinant, by the folklore connection between space and depth complexity.

To recap, the main logical steps of the Taylor series meta-algorithm are:
\begin{enumerate}
\item
Define a polynomial $g(z)$ that interpolates between an easy-to-compute determinant at $z=0$
and the determinant of our target matrix $A$.
\item
Choose the polynomial $g(z)$ so that for a large natural class of matrices we have that $g(z)$ satisfies
the condition that all its roots are bounded-away from the unit disk.
\item
Demonstrate a low-depth algorithm for computing all low-order derivatives of $g(z)$.
\end{enumerate}
Notably, the first two steps all appeared in works on the permanent \cite{Barvinok2016, Bar16, EM18}.
However, the third step is required for the case of the determinant, where our goal is to construct
a {\it low-depth} circuit.

In this work, we solve these requirements in the following way:
\begin{enumerate}
\item
We set $g(z) = \Det( (1-z) \cdot I + z \cdot A)$. Clearly $g(0) = \Det(I)$ which is easy to compute and $g(1) = \Det(A)$.
\item
The polynomial $g(z)$ is reminiscent of the characteristic function of $A$
$$
\chi(A) \equiv \Det( \lambda I  - A)
$$
One can easily check that if $A$ is a Hermitian matrix that is well-conditioned the roots of $g(z)$ are all bounded away
from either $z=0$ or $z=1$, {and that they are all real}.
\item
In order to compute the derivatives of $g(z)$ using shallow circuits we build upon the fact that
(contrary to the permanent) we do in fact have a polynomial-time algorithm for the determinant.
We use that algorithm, in conjunction with the algorithm for parallelizing computation of low-degree polynomials due to Valiant et al. \cite{valiant1983fast},
to show that any order-$t$ derivative of $g(z)$ can be computed by a circuit of depth $O(\log(n) \cdot \log (t))$.
\end{enumerate}

In order to compute a multiplicative approximation of the determinant of the input matrix, several additional steps must be added
that can compute the derivatives of $f(z) = \log g(z)$ from those of $g(z)$, and making sure that one can implement
the arithmetic circuits for these polynomials using Boolean circuits with small overhead (which is one of the reasons that our space is not precisely logarithmic but rather has extra $\log\log(n)$ factors.)
We summarize the main steps of the parallel algorithm here and refer the reader to section \ref{sec:alg} for a detailed
description of the computational steps.
\begin{mdframed}
\begin{enumerate}
\item
Input: $\kappa \leq \poly(n)$, and an $n\times n$ Hermitian or Hurwitz stable matrix $A$ such that $I / \kappa \preceq |A| \preceq I$.
\item
Round each entry to $O(\kappa \log(n))$ bits of precision.
\item
Compute the first $k = (\log n) \cdot (\kappa \log \log n)^{O(\log \log \kappa )}$ derivatives of $g(z) = \Det( (1-z) I + z A)$ at $z=0$
using a dynamic program that is attached to the Samuelson-Berkowitz algorithm \cite{B84, S42}.
This dynamic program can be parallelized to depth $\tilde{O}(\log(n) \cdot \log(k)) = \tilde{O}(\log(n) \cdot \log(\kappa))$ by the  algorithm for parallelizing low-degree arithmetic circuits due to Valiant-Skyum-Berkowitz-Rackoff \cite{valiant1983fast}.
\item
Using Bell polynomials, compute the first $k$ derivatives of $f(z) = \log(g(z))$ at $z=0$ as in \cite{Bar16}. Also parallelize this step to depth $\tilde{O}(\log(n) \cdot \log(\kappa))$ using \cite{valiant1983fast}.
\item
Use CAC interpolation, introduced in \cite{EM18}, to compute the value of $f(1)$ by constructing an interpolation path that avoids the roots
of $g(z)$ (or poles of $f(z))$. Again parallelize CAC interpolation to depth $\tilde{O}(\log(n) \cdot \log(\kappa))$ using \cite{valiant1983fast}.
\item
Return $e^{f(1)}$.
\end{enumerate}
\end{mdframed}

\subsection{Discussion and Future Work}

Our result implies that the determinant of a large class of matrices, namely polylogarithmically well-conditioned Hermitian or Hurwitz stable matrices, can be approximated to inverse polynomial
relative error in space which is nearly logarithmic in the matrix size.
A natural next step would be to try to extend our algorithm to run in logarithmic space or depth for matrices with
polynomial condition number, which could then amount to an $\NC^1$ algorithm for "almost" any matrix in the Wigner ensemble \cite{Tao}. Another direction could be to try to show that approximating the determinants of polynomially-conditioned matrices is $\BQL$-complete, as discussed in Section \ref{subsubsec:approxdetcomputation}. We note that proving both the algorithm and the completeness result would imply the de-quantization of $\BQL$.

\section{Preliminaries}

\subsection{Basics}

Given a complex matrix $A \in \mathbb{C}^{n \times n}$, let $A^{\dag}$ denote its conjugate transpose. $A$ is Hermitian if $A = A^{\dag}$, in which case the eigenvalues of $A$ are real. $A$ is positive semi-definite (PSD) if it is Hermitian and has nonnegative eigenvalues. We write $A \succeq 0$ if $A$ is PSD. For Hermitian matrices $A$ and $B$, we write $A \succeq B$ if $A - B \succeq 0$, and we note that $\succeq$ defines a valid partial ordering. The absolute value of $A$ is defined as $|A| = \sqrt{A^{\dag}A}$. The singular values $0\leq s_n(A) \leq \dots \leq s_1(A)$ of $A$ are the eigenvalues $s_i(A) = \lambda_i(|A|)$ of $|A|$. The spectral norm $\|A\|_2$, is the maximum singular value $s_1(A)$. The max-norm $\|A\|_{\max} = \max_{i,j} |A_{i,j}|$ is the maximum absolute value of an entry in $A$.
\begin{definition}[Condition number]
The condition number of $A$ is $\kappa(A) := s_1(A) / s_n(A)$.
\end{definition}
In this paper, we will focus on well-conditioned Hermitian and Hurwitz stable matrices that are normalized to have spectral norm at most $1$:
\begin{definition} [Well-conditioned matrices] Let $0 \leq \delta \leq 1$. 
\begin{enumerate}
\item[(1)] The class of well-conditioned Hermitian matrices with parameter $\delta$ is defined as 
$$
{\cal H}_\delta = 
\{A \in \C^{n\times n} : A = A^\dag, \delta \cdot I \preceq |A| \preceq I \}.
$$
\item[(2)] The class of well-conditioned Hurwitz stable matrices with parameter $\delta$ is defined as 
$$
{\cal S}_\delta = 
\{A \in \C^{n\times n} : \forall i, \Re(\lambda_i(A)) < 0, \delta \cdot I \preceq |A| \preceq I \}.
$$
\end{enumerate}
Note that a matrix $A \in {\cal H}_{\delta} \cup {\cal S}_{\delta}$ has condition number $\kappa \leq 1/\delta$.
\end{definition}
%

One of the main complexity classes discussed in this paper is $\NC$ defined as the following.
\begin{definition} [Nick's class]  $\NC [h (n)]$ is the class of $\{0,1\}^n \rightarrow \{0,1\}^*$ Boolean functions computable by a logspace-uniform family of Boolean circuits, with polynomial size, 
depth $O(h)$, and fan-in $2$. $\NC^i := \NC [\log^i n]$.
\end{definition}

\subsection{Arithmetic circuits}
\begin{definition}[Polynomials]
\noindent
Let $g \in \bbF [x_1, \ldots, x_N]$, be a multivariate polynomial in variables $x_1, \ldots, x_N$, over field $\bbF$. The degree of a monomial of $g$ is the total number of variables in the monomial (counting with multiplicity). The total degree of $g$ ($\deg g$) is the maximum degree of a monomial in $g$. 
\end{definition}

\begin{definition}[Arithmetic circuits] An arithmetic circuit is a directed acyclic graph on nodes $v_1,\ldots,v_s$, called gates. If a node has indegree 0, it is called an input gate, and is labelled with a either a field element in $\mathbb{F}$ or a variable in $\{x_1,\ldots,x_N\}$. Otherwise $v$ is labelled as either an addition or a multiplication gate. Finally, $v_s$ is the ``output'' gate of the circuit.

Each gate $v$ recursively computes a polynomial $p_v \in\bbF[x_1,\ldots,x_N]$. If $v$ is an input gate, then $p_v$ is its label. Otherwise, let $w_1,\ldots,w_t$ be the children of $v$. If $v$ is an addition gate then it computes $p_v = \sum_{i=1}^t p_{w_i}$, and if $v$ is a multiplication gate then it computes $p_v = \prod_{i=1}^t p_{w_i}$. Overall, the arithmetic circuit is said to compute the polynomial $p_{v_s}$ computed at its output gate.

If all gates have indegree 0 or indegree 2, then the circuit is said to have fan-in 2. Except when explicitly stated otherwise, all arithmetic circuits in this paper have fan-in 2.
\label{def:circuits}
\end{definition}

In this paper, we will use two fundamental algorithms from the arithmetic circuit literature. The first algorithm, which can be traced back to Strassen \cite{strassen73vermeidung}, allows us to efficiently convert an arithmetic circuit ${\cal C}$ computing a polynomial $p(x_1,\ldots,x_N,z) \in \bbF[x_1,\ldots,x_N,z]$ into an arithmetic circuit ${\cal C}'$ computing the coefficient of $z^k$ in $p$ (which itself is a polynomial in $\bbF[x_1,\ldots,x_N]$). Formally:

\begin{definition}For any polynomial $g(z,x_1,\ldots,x_N) \in \bbF[z,x_1,\ldots,x_N]$ and integer $k \geq 0$, let 
$$
[z^k]g(z,x_1,\ldots,x_N) \in \bbF[x_1,\ldots,x_N]
$$ 
denote the coefficient of $z^k$ in $g$: i.e., 
$$
g(z,x_1,\hdots, x_N) = \sum_{i=0}^{\infty} ([z^k]g(z,x_1,\ldots,x_N)) \cdot z^k
$$
\end{definition}
Note that up to a factor of $k!$, the polynomial $[z^k]g$ is the same as the $k$th partial derivative of $g$ with respect to $z$, evaluated at $z = 0$:
$$k! \cdot [z^k]g(z,x_1,\ldots,x_N) = \frac{\partial^k}{\partial z^k} g(z,x_1,\ldots,x_N) |_{z=0}.$$

The result that we use is stated below.
\begin{lemma}[Computing the single-variable derivative of an arithmetic circuit, \cite{strassen73vermeidung}]\label{lem:deriv}
Let 
$$
g(z, x_1, \ldots, x_N) \in \bbF[z,x_1,\ldots,x_N]
$$ 
be a polynomial computed by a fan-in-2 arithmetic circuit ${\cal C}$ of height $h$. 
Then, for any $k\geq 0$, $[z^k]g(z,x_1,\ldots,x_N)$ can be computed by a fan-in-2 arithmetic circuit ${\cal C}'$ of size $|{\cal C}'| = O(k^2 \cdot |{\cal C}|)$ and depth $(k+1)h$. Moreover, ${\cal C}'$ can be computed from ${\cal C}$ and $k$ in logarithmic space.
\end{lemma}

\begin{proof} For each node $v$ of the circuit ${\cal C}$, let $p_v(x_1,\hdots, x_N, z)$ denote the polynomial computed at $v$. We construct a circuit ${\cal C}'$ computing $[z^k]$ with a dynamic program based on ${\cal C}$:
\begin{enumerate}
\item For each gate $v$ of ${\cal C}$ and each integer $0 \leq i \leq k$, add a gate $(v,i)$ to $C'$. We will guarantee that the polynomial $p'_{(v,i)}$ computed at $(v,i)$ equals $[z^i]p_v(z,x_1,\ldots,x_N)$.
\item For each $(v,i)$ such that $v$ is an input gate of ${\cal C}$, let $(v,i)$ be an input gate of ${\cal C}'$ and label it by $p'_{(v,i)} = [z^i]p_v \in \bbF \cup \{x_1,\ldots,x_N\}$.
\item For each $(v,i)$ such that $v$ is an addition gate of ${\cal C}$ with inputs $w_1,w_2$, let $p'_{(v,i)} = p'_{(w_1,i)} + p'_{(w_2,i)}$.
\item For each $(v,i)$ such that $v$ is a multiplication gate of ${\cal C}$ with inputs $w_1,w_2$, let $p'_{(v,i)} = \sum_{j=0}^i p'_{(w_1,j)} \cdot p'_{(w_2,i-j)}$. This can be implemented by adding at most $2i \leq 2k$ intermediate addition and multiplication gates.
\end{enumerate} By induction on the depth, the polynomial $p'_{(v,i)}$ computed at each gate $(v,i)$ equals $[z^i]p_v$. Let the output gate of ${\cal C}'$ be $(v_s,k)$, where $v_s$ is the output gate of ${\cal C}$. Therefore ${\cal C}'$ correctly computes $[z^k]p_{v_s}$. This entire construction can be implemented in logarithmic space.

Finally, $|{\cal C}'| = O(k^2 |{\cal C}|)$ because for each gate $v$ in ${\cal C}$ at most $2k(k+1)$ gates are added in the construction of ${\cal C}'$. And ${\cal C'}$ has depth $(k+1)h$ because each gate is replaced with a gadget of height at most $k+1$.
\end{proof}


The second classical result that we require is the depth-reduction theorem of Valiant-Skyum-Berkowitz-Rackoff for fast parallel computation of low-degree polynomials. Informally, this theorem states that if a low-degree polynomial can be computed by a small arithmetic circuit, then it can also be computed in low parallel complexity:
\begin{theorem}[Depth reduction for arithmetic circuits, \cite{valiant1983fast}]
Let $g(x_1,\ldots,x_N) \in \bbF[x_1,\ldots,x_N]$ be a multivariate polynomial of total degree $d$ computed by a fan-in-2 arithmetic circuit ${\cal C}$ of size $s$. Then there is an arithmetic circuit $D({\cal C})$ of size $\poly(sd)$ and depth 
$O(\log d)$ that computes $g$. Moreover, $D({\cal C})$ can be computed from ${\cal C}$ in logarithmic space, each multiplication gate of $D({\cal C})$ has fan-in 2, and each addition gate of $D({\cal C})$ has fan-in $\poly(sd)$.
\label{thm:depth-reduction}

In particular, by replacing each addition gate by a $O(\log(sd))$-depth tree of fan-in 2 addition gates, $D(\mathcal{C})$ can be transformed into a $O((\log d) \cdot (\log s + \log d))$-depth arithmetic circuit of size $\poly(sd)$ and fan-in 2.
\end{theorem}

Let us illustrate this result with an example application to the exact computation of the determinant. We know that the determinant $\Det(A) = \sum_{\sigma \in S_n} \prod_{i \in [n]} \sgn(\sigma) A_{i\sigma(i)}$ is a degree-$n$ polynomial in the entries of $A$, and that it has an arithmetic circuit of size $\poly(n)$\footnote{For example, this small circuit can be constructed from a division-free variant of Gaussian elimination.}.  Therefore, Theorem \ref{thm:depth-reduction} implies that there is a $O(\poly(sd)) = O(\poly(n))$-size and $O((\log d) \cdot (\log sd)) = O((\log n)^2)$-depth arithmetic circuit computing $\Det(A)$. This result was mentioned in the introduction.

An $O((\log n)^2)$-depth circuit for the exact computation of $\Det(A)$ is not sufficient for our purposes. Recall that our goal is instead to multiplicatively approximate $\Det(A)$ using depth $\tilde{O}((\log n) \cdot (\log \kappa))$, which scales with the condition number $\kappa$ of $A$. Hence, when $A$ is particularly well-conditioned (e.g., $\kappa = O(\poly\log(n))$), then our circuit will have $o((\log n)^2)$ depth. At a high level, in order to achieve this result we will also apply Theorem \ref{thm:depth-reduction}. However, instead of applying the theorem directly to $\Det(A)$ we will apply it to a $\poly(n)$-size-computable degree-$O(\poly(\kappa))$ polynomial that approximates $\Det(A)$. Assuming without loss of generality that $\kappa = O(\poly(n))$, this will give a $O((\log n) \cdot (\log \kappa))$-depth arithmetic circuit for the polynomial approximating $A$.

\subsection{From arithmetic circuits to Boolean circuits}

In this paper we use arithmetic circuits to represent and manipulate low-degree polynomials over $\bbC$. In order to evaluate low-depth arithmetic circuits, we have to translate them into low-depth Boolean circuits. This takes some care, because, when we convert arithmetic circuits into Boolean circuits, we cannot represent the values computed at each gate up to arbitrary precision.

Our approach is to Booleanize an arithmetic circuit $\mathcal{C}$ on variables $x_1,\ldots,x_N$ by rounding every input to $r$ bits of precision and then replacing each arithmetic operation in the circuit with the corresponding exact Boolean implementation, assuming that the inputs $x$ are such that $\max_i |x_i|$ is bounded by some number $M$. The resulting Boolean circuit is denoted by $B_{r,M}({\cal C})$.  In order to ensure that $B_{r,M}({\cal C})$ remains small and low-depth, we have to bound the number of bits used to represent the intermediate values in the computation. In order to ensure that $B_{r,M}({\cal C})$ is accurate, we also have to bound the error incurred by the rounding step. These correctness guarantees are provided by the following lemma:
\begin{lemma} \label{lem:Booleanization}
Let $\eps > 0$, and let ${\cal C}$ be a circuit over $\bbC$ of depth $h$, computing a polynomial $g(x_1,\ldots,x_N)$ of degree $d$. Suppose that each multiplication gate of ${\cal C}$ is of fan-in 2, and each addition gate is of fan-in at most $m$. For technical reasons, suppose that all input gates of ${\cal C}$ are labelled by a variable in $\{x_1,\ldots,x_N\}$ (i.e., there are no input gates labelled by a constant in $\bbC$).

If $r > (2hd^2 \ceil{\log(m)} + 1) \log_2(4NMd/\eps)$, then $B_{r,M}({\cal C})$ is a logspace-uniform Boolean circuit of size $\poly(|\mathcal{C}|dhr(\log m)\log(M))$ and depth $O(h \cdot \log(dhrmM))$. Moreover, $B_{r,M}({\cal C})$ computes a function $\tilde{g}(x_1,\ldots,x_N)$ such that for all $a_1,\ldots,a_N \in \bbC$ with $\max_i |x_i| \leq M$, $$|\tilde{g}(a_1,\ldots,a_n) - g(a_1,\ldots,g_n)| < \eps.$$
\end{lemma}
Note that Lemma \ref{lem:Booleanization} requires that each of the input gates of ${\cal C}$ be labelled with an input variable: in other words, none of the input gate labels are constants from $\mathbb{C}$. We place this technical restriction so that we can conveniently bound the bit complexity of the values computed by the circuit. This is not an important restriction in our case, because for all of the arithmetic circuits considered in this paper, the degree of the polynomial computed by the circuit does not significantly change if we replace each input gate constant $c \in \bbC$ with a variable $y_c$ whose value will eventually be hard-coded to $c$. The proof of Lemma \ref{lem:Booleanization} is deferred to Appendix \ref{app:Booleanizationdetails}.

\section{Determinants and Complex Polynomials}

The determinant of an $n\times n$ matrix can be computed efficiently by a well-known result of Samuelson and Berkowitz:
\begin{theorem}[Samuelson-Berkowitz \cite{B84, S42}]\label{thm:det1} 
The determinant of an $n\times n$ matrix can be computed by an arithmetic circuit of size $\poly(n)$ and fan-in 2.
\end{theorem}

Similarly to the line of work pioneered by Barvinok \cite{Bar16}, we analyze this problem using tools for analyzing
complex polynomials:
\begin{definition}[Disks, sleeves and root-free area]

For a polynomial $g: \C \rightarrow \C$ and $S \subseteq \C$, $g$ is root free inside $S$ if $z \in S \implies g(z) \neq 0$.
In this work we will use three specific kinds of regions $S$:
\begin{enumerate}
    \item 
    Open disk: denote an open disk of radius $r \geq 0$ around $c \in \C$ by ${\cal D} (c,r) = \{ z \in \C : |z - c| < r\}$.
    \item
    Unit sleeve:
    The unit sleeve with width $\delta$ is denoted with ${\cal S}_\delta := \{z  \in \C : |\Im (z)| < \delta \text { and } -\delta < \Re (z) < 1+ \delta\}$. 
    \item Half-plane:
    The left-hand side of a point $x \in \bbR$ is defined as ${\cal P}_x := \{z  \in \C : \Re (z) < x \}$.  
\end{enumerate}
   
  \label{def:disk-sleeve}
\end{definition}

\subsection{The determinant interpolation polynomial}

Let $g_A(z) = \Det (I(1-z) + z A)$. Therefore
\begin{align*}
g_A(0) &= 1,\\
g_A(1) &= \Det(A).
\end{align*}

\begin{lemma} [Derivatives of $g_A$]
The $k$-th derivative of $g_{A}(z) = \Det(I(1-z) + z A)$
at $z=0$
is a polynomial of degree $k$ in the entries of $A$.
\label{lem:der-degree}
\end{lemma}
\begin{proof}
We use the notation $B \rightarrow_k A$ to denote $B$ is a $k\times k$ principal sub-matrix of $A$. We show that the $k$-th derivative of $g_A(z)$ at $z=0$ is $g^{(k)}(0) = \sum_{B \rightarrow_k A} \Det (B- I)$:
\begin{align*}
g^{(k)}(0) &= \sum_{\sigma \in S_n} \sgn(\sigma) \sum_{i_1 < \ldots < i_k} \prod_{j\notin \{i_1, \ldots, i_k\}}\delta_{j, \sigma(j)} \prod_{j\in \{i_1, \ldots, i_k\}}(a_{j,\sigma(j)} -\delta_{j, \sigma(j)})\\
&= \sum_{i_1 < \ldots < i_k} \sum_{\sigma \in S \{i_1, \ldots, i_k\}} \sgn(\sigma)   \prod_{j\in \{i_1, \ldots, i_k\}}(a_{j,\sigma(j)} -\delta_{j, \sigma(j)})\\
&=  \sum_{B \rightarrow_k A} \Det (B- I).
\end{align*}
Each summand $\Det (B- I)$ is a polynomial of degree $k$ in the entries of $B$ and hence degree $k$ in the entries of $A$.
\end{proof}

\begin{theorem}[Roots vs. eigenvalues] 
Let $0 < \delta < 1$.
\begin{enumerate}
\item [(1)] (Hermitian) If $A \in {\cal H}_\delta$ then $g_A$ is root-free inside ${\cal D} (0,1/2) \cup {\cal D} (1, \frac \delta {1+\delta})$.
\item [(2)] (Hurwitz stable) If $A \in {\cal S}_\delta$ then $g_{-A}$ is root-free inside ${\cal P}_{1/2} \cup {\cal D} (1/2,1/2) \cup {\cal D} (1, \frac \delta {1+\delta})$.
\end{enumerate}
  \label{thm:roots}
\end{theorem}

\begin{proof}
Let $\omega_1, \ldots, \omega_n$ be the eigenvalues of $A$.
Then
$$
g_A(z) = \Det (I(1-z) + z A) = \prod_{i \in [n]} ((1-z) + z \omega_i).
$$
For any $\omega_i \neq 1$, $z_i := \frac 1 {1- \omega_i}$ is a root of $g_A$. 
Also if $A$ has a unit eigenvalue that does not introduce a root for $g_A$. 

\begin{enumerate}
\item [(1)] First, 
since
$-1 \leq \omega_i$ implies that $z_i \geq 1/2$ this establishes ${\cal D} (0,1/2)$ as a root-free disk. 
Second, $\delta \cdot I \preceq |A|$
 implies $\forall i \in [n], \quad |\omega_i| \geq \delta$ which implies $\forall i \in [n]$, either $z_i \geq \frac 1 {1-\delta}$ or $z_i \leq \frac 1 {1+\delta}$. This establishes ${\cal D} (1, \frac \delta {1+\delta})$ as a root-free disk.
 
 \item [(2)] When $A \in {\cal S}_\delta$ the eigenvalues of $-A$ lie inside $\Lambda = \{\omega \in \C : \delta \leq  |\omega| \leq 1, \Re(\omega) > 0\}$. We compute the image of $\Lambda$ under the map $z = \frac{1}{1-\omega}$ through the following observations: 
\begin{enumerate}
\item [(i)] $ |\omega| \geq \delta \Leftrightarrow |z-1| \geq \delta \cdot |z| \Rightarrow |z-1| \geq \delta (1- |1-z|) \Leftrightarrow |1-z| \geq \frac{\delta}{1+\delta},$
\item [(ii)]  $|\omega| \leq 1 \Leftrightarrow |z-1| \leq |z| \Leftrightarrow \Re(z) \geq 1/2,$
\item [(iii)]  $\Re(\omega) \geq 0 \Leftrightarrow |1 + \omega | \geq |1-\omega| \Leftrightarrow |2 z - 1| \geq 1 \Leftrightarrow |z-1/2| \geq 1/2$.
\end{enumerate}
Items (i), (ii) and (iii) establish root-freeness inside ${\cal D} (1, \frac \delta {1+\delta})$, ${\cal P}_{1/2}$ and ${\cal D} (1/2,1/2)$, respectively.
\end{enumerate}
\end{proof}


\section{Computational Analytic-Continuation}

\subsection{Improved analysis for CAC}

In \cite{EM18} a subset of the authors has outlined an algorithmic technique called CAC for interpolating the value
of a complex function given sufficiently many derivatives at some origin.
In this work, we require slightly stronger conditions on the performance of the algorithm so we present
a refined analysis thereof.
We begin by rewriting the algorithm with slightly modified parameters:

\begin{mdframed}

\begin{algorithm}[Computational analytic continuation]
\label{algorithm:cac}
\begin{enumerate}

\noindent
\item
\textbf{Input:} 
Integer $m_0 \geq 1$.
An oracle $\mathcal{O}_g$ that takes a number $m_0$ 
as input and outputs the first $m_0$ derivatives of $g$ at $z = 0$, where $g(0) = 1$.
$t$ complex numbers $\Delta_1,\hdots,\Delta_t$. 
A number $0 < \theta \leq 1$.
\item
\textbf{Fixed parameters:}

\begin{enumerate}
\item 
$s_0=0$ and $s_i= s_0 + \sum_{j=1}^{i} \Delta_j $ 
for each $1 \leq i \leq t$ \hfill 
\end{enumerate}

\item
\textbf{Variables:}

\begin{enumerate}
\item $m_i$ for $0\leq i \leq t$ \hfill \% the number of derivatives at each point $s_i$.
\item $\hat{f}^{(l)}_i$ for $0 \leq l\leq m_i$ and $0 \leq i \leq t$ \hfill \% the $l$'th derivative of $f$ at $s_i$.
\end{enumerate}

\item
\textbf{Main:}
\begin{enumerate}
\item\label{it:it1} Query $\mathcal{O}_g(m_0)$ to obtain $g^{(0)}(0),\ldots, g^{(m_0)}(0)$

\item\label{it:it2} Using derivatives from step \ref{it:it1} 
compute 
$\hat{f}^{(l)}_0 \leftarrow f^{(l)}(s_0)$ 
for $1 \leq l\leq m$. 

\item\label{it:iteration} For each $i = 0,\ldots,t-1$:

\begin{itemize}
\item Set: $m_{i+1}  = \ceil{\theta m_i / (2\log\  m_i)}$.
\item Compute
$
\forall 0 \leq l \leq m_{i+1}, \ \ \ 
\hat{f}^{(l)}_{i+1}  = \sum_{p = 0}^{m_{i} - l} \frac{\hat{f}^{(p + l)}_i}{p!} \Delta_i^{p}.
$
\end{itemize}
\label{it:dyn}
\end{enumerate}

\item
\textbf{Output}:

Let $\hat{f} := \hat{f}^{(0)}_t$ and return ${\cal O} = e^{\hat{f}}$.

\end{enumerate}
\vspace{5mm}
\end{algorithm}

\end{mdframed}

%

\noindent

\begin{lemma}
[Correctness of algorithm \ref{algorithm:cac}]\label{lem:caccorrectness}

\noindent
Let $g(z)$ be a polynomial of degree at most $n$ such that $g(0) = 1$, and let $f(z) = \log(g(z))$.
Suppose the inputs to algorithm \ref{algorithm:cac} satisfy the following conditions:
\begin{enumerate}
\item
Non-increasing sequence of segments: $|\Delta_i| \leq |\Delta_{i-1}|$ for all $i \geq 1$.
\item\label{it:ratio} Root avoidance: 
For each $i$ the ratio between the distance from the closest root of $g(z)$ to $s_i$ and the step size $|\Delta_{i+1}|$ is at least $\beta = e^{\theta}$ for $0 < \theta \leq 1$.
\end{enumerate}
Then, for small enough $\epsilon > 0$, letting \be \label{eq:m0setting} m_0 \geq 10\left(\log(n/\epsilon \theta)\right)\left(10t(\log t + \log \log(n/\epsilon \theta))\right)^{t},\ee
Algorithm \ref{algorithm:cac} outputs an $\eps$-additive approximation to $f(s_t)$.
\end{lemma}

%
%

Prior to establishing the correctness of the algorithm, we define shifted versions of $g(z)$ as follows:
\be
\forall i\in [t] \quad
\tilde{g}_i(z) = g(z + s_i),
\label{eq:tilde-functions}
\ee
and
\be
\tilde{f}_i(z) = \log(\tilde{g}_i(z)),
\ee
and denote $f_i^{(l)} = \tilde{f}_i^{(l)}(0)$.
We need the following elementary fact, which we leave without proof:
\begin{lemma}
If the closest root of $g$ to the point $s_i$ in the complex plane is $\lambda$, then the closest root of $\tilde{g}_i$ to $z=0$ is also $\lambda$.
\label{lemma:closest}
\end{lemma}

\noindent
We now prove correctness of the algorithm:

\begin{proof}[Proof of Lemma \ref{lem:caccorrectness}]
Let $f (z) := \log (g(z))$. It is sufficient to show that 
\be\label{eq:delta}
\left|\hat{f}-f(s_t)\right| \leq \eps
\ee


 Let $\hat{f}^{(k)}_i$ denote the approximation of the $k$'th derivative of $f$ at point $s_i$ obtained by the algorithm.
Using oracle $\mathcal{O}_g$ for $0 \leq l \leq m_0$ 
we can compute precisely the derivatives of $g$ at $s_0=0$ 
and using Lemma \ref{lem:brunodifaasubstitute} (whose statement and proof we momentarily defer) we can compute the derivatives of $f$ precisely at $s_0$:
\be
\hat{f}^{(l)}_0 \leftarrow f^{(l)} (s_0).
\ee
For $i = 1,\ldots,t$ (in order) 
algorithm \ref{algorithm:cac} computes the lowest $m_i$ derivatives at $s_i$ using the first $m_{i-1}$ derivatives at $s_{i-1}$ as follows:
\be\label{eq:alg1}
\forall 0 \leq l \leq m_i, \quad
\hat{f}^{(l)}_i  = \sum_{p = 0}^{m_{i-1} - l} \frac{\hat{f}^{(p + l)}_{i-1}}{p!} \Delta_i^{p}.
\ee
By assumption \ref{it:ratio} and Lemma \ref{lemma:closest} for each $1\leq i \leq t$ the function $\tilde{f}_{i-1}$ is analytical about point 
$0$ in a disk of radius $\beta |\Delta_i|$.
Hence, we can write the $\ell$-th derivative of $\tilde{f}_i(z)$ 
as the infinite Taylor-series expansion of the $\ell$-th derivative of $\tilde{f}_{i-1}(z)$ evaluated at point $\Delta_i$:
\be\label{eq:ideal1}
\tilde f_{i}^{(\ell)} := \tilde f^{(l)}_{i} (0) = \sum_{p=0}^{\infty} \frac{ \tilde f^{(p+l)}_{i-1} (0)}{ p!} \Delta_i^p.
\ee
Let ${\cal E}^{(l)}_i$ denote the additive approximation error of the $l$-th derivative at step $i\in [t]$
 and $0 \leq l \leq m_i$.
\be
\mathcal{E}^{(l)}_i := \Big| \hat{f}^{(l)}_i - {f}^{(l)} _{i}\Big|, \hspace{1cm}  \forall 0 \leq l \leq m_{i}
\ee
Using the triangle inequality to bound the difference between equations \eqref{eq:alg1} and \eqref{eq:ideal1}, we
get:
\bea 
\forall i\in [t],  0 \leq l \leq m_i,\ \ 
\mathcal{E}^{(l)}_i &\leq& 
\sum_{p=0}^{m_{i-1} - l} \frac{|\hat{f}^{(p+l)} _{i-1}-{\tilde f}^{(p+l)} _{i-1}|}{p!} |\Delta_i|^{p} + 
\sum_{p=m_{i-1} - l +1}^{\infty} \frac{|{\tilde f}^{(p+l)} _{i-1} |}{p!} |\Delta_i|^{p},\\
&=& \sum_{p=0}^{m_{i-1} - l} \frac{\mathcal{E}^{(p+l)}_{i-1}}{p!} |\Delta_{i}|^{p} + 
\sum_{p=m_{i-1} - l +1}^{\infty} \frac{|{\tilde f}^{(l+p)}_{i-1} |}{p!} |\Delta_{i}|^p,\\
&=:&\sum_{p=0}^{m_{i-1} - l} \frac{\mathcal{E}^{(p+l)}_{i-1}}{p!} |\Delta_{i}|^{p} +\kappa_{i,l},
\label{equation:err}
\eea
where
\be
\kappa_{i,l} := \sum_{p=m_{i-1} - l +1}^{\infty} \frac{|{\tilde{f}}^{(p+l)}_{i-1} |}{p!} |\Delta_i|^{p} = 
\sum_{p=m_{i-1} - l +1}^{\infty} \frac{|\tilde{f}_{i-1}^{(p+l)} (0) |}{p!} |\Delta_i|^{p}. 
\label{eq:kappa4}
\ee
At this point, we focus on placing an upper bound on $\kappa_{i,l}$. Fix any index $i$ and
let $z_1, \ldots, z_n$ be the roots of the shifted function $\tilde{g}_{i-1}$.
Then
\be
\tilde{g}_{i-1} (z) = \tilde{g}_{i-1} (0) \left(1- \frac{z}{z_1}\right) \ldots \left(1- \frac{z}{z_n}\right). 
\ee
We can write the derivatives of $\tilde{f}_{i-1}(z) = \log(\tilde{g}_{i-1}(z))$ in terms of $z_1,\ldots,z_n$:
\be
\forall k > 0, \quad
\tilde{f}^{(k)}_{i-1} (0) = - \sum_{j=1}^n\frac{(k-1)!}{z^k_j}. 
\ee
Using these derivatives and the triangle inequality we can bound equation \eqref{eq:kappa4} for all $0 \leq l \leq m_{i}$,
\ba\label{eq:kappa5}
\kappa_{i,l} &\leq \sum_{j=1}^n \sum_{p=m_{i-1} - l +1}^{\infty} \frac{(l + p-1)!}{p!} \frac{|\Delta_i|^{p}}{|z_j|^{p + l}} \\
&\leq e \cdot \sum_{j=1}^n \sum_{p=m_{i-1} - l +1}^{\infty} (p/e)^{l} (\frac {l + p}{p})^{p+l} \frac{|\Delta_i|^{p}}{|z_j|^{p + l}} & \text{using Lemma \ref{lem:factorialdivbound}}\\
&\leq \frac{en}{|\Delta_{i}|^l} \sum_{p=m_{i-1} - l +1}^{\infty} (p/e)^{l}  (\frac {l + p}p)^{p+l} \frac{1}{\beta^{p+l}}\\
&\leq \frac{en}{|\Delta_{i}|^l} (\frac {m_{i-1} + 1}{m_{i-1}-l + 1})^{l}\sum_{p=m_{i-1} - l +1}^{\infty} (1+l/p)^p (p/e)^{l}  \frac{1}{\beta^{p+l}}\\
&\leq \frac{en}{|\Delta_{i}|^l} (\frac {m_{i-1} + 1}{m_{i-1}-l + 1})^{l}\sum_{p=m_{i-1} - l +1}^{\infty} p^{l}  \frac{1}{\beta^{p+l}}  \label{eq:beforeint1}
\ea

In order to bound this quantity, we prove a lower-bound on $m_i$ for all $0 \leq i \leq t$. Since the update rule for $m_i$ is $m_{i+1} = \ceil{\theta m_i / (2 \log m_i)}$, and $x / \log(x)$ is increasing for $x > 10$, in order to prove the lower bound on $m_i$ we can without loss of generality assume
$m_0 = \ceil{10\left(\log(n/\epsilon \theta)\right)\left(10t(\log t + \log \log(n/\epsilon \theta))\right)^{t}}$. The following facts immediately follow.
For all $0 \leq i < t$, \ba \label{eq:logmiupperbound}\log(m_i) &\leq \log(m_0) \\ & \leq \log(11) + \log\log(n/\epsilon \theta) + t\log(10t) + t \log\log\log(n/\epsilon \theta) + t \log \log(t) \\ &\leq 2t\log(t) + 2t \log\log(n/\epsilon \theta), \hspace{5.5cm},\ea and therefore for all $0 \leq i \leq t$ \ba m_i &\geq 10 (\log(n/\epsilon \theta))(10t(\log t + \log \log(n/\epsilon \theta)))^{t-i} 5^i\ea and in particular \ba m_i &\geq 10 (\log(n/\epsilon \theta)) \cdot 5^t \\ &\geq 10(t + \log(n/\epsilon \theta)) \label{eq:milowerbound} .\ea So \ba m_{i-1} &= 2 m_i \log(m_{i-1}) / \theta & \text{by construction} \\ &\geq (m_i \log(m_{i-1}) + m_i) / \theta \\ &\geq (m_i \log(m_{i-1}) + 10t + 10\log(n/\epsilon \theta))/\theta & \text{using } \eqref{eq:milowerbound} \\
 &\geq (m_i \log(m_{i-1}) + \log(e^3ne^t/\epsilon \theta))/\theta \label{eq:midecreasesfastenough}\ea

Also, by the lower bound \eqref{eq:milowerbound} on $m_{i-1}$ the algorithm chooses $m_{i-1} \geq 3m_i / \theta \geq m_i \cdot (1 + 2/\theta)$, so since $l \leq m_i$ it follows that \begin{equation} \label{eq:michange} m_{i-1} - l + 1 > 2l/\theta.\end{equation} Therefore we may apply the bound of technical lemma \ref{lem:int1} to \eqref{eq:beforeint1},
\ba
\kappa_{i,l} &\leq \frac{en}{|\Delta_{i}|^l} (\frac {m_{i-1} + 1}{m_{i-1}-l + 1})^{l}\frac{(m_{i-1} - l + 1)^l}{\beta^{m_{i-1} +1}(1-\beta^{-1}e^{l/(m_{i-1}-l+1)})} & \text{using Lemma  \ref{lem:int1}} \label{eq:afterint1}\\
&\leq \frac{en}{|\Delta_{i}|^l} \frac{(m_{i-1} + 1)^l}{\beta^{m_{i-1} +1}(1-\beta^{-1/2})} & \text{using \eqref{eq:michange}} 
\\
&\leq \frac{e^2n}{\theta |\Delta_{i}|^l} \frac{(m_{i-1} + 1)^l}{e^{\theta m_{i-1}}} & \text{using } \beta = e^{\theta}, 0 < \theta \leq 1 \\
&\leq \frac{e^2n}{\theta |\Delta_i|^{l}} \frac{(m_{i-1} + 1)^{m_i}}{e^{\theta m_{i-1}}} & \text{using }l \leq m_i \\
&\leq \frac{e^3n}{\theta |\Delta_i|^l} \frac{(m_{i-1})^{m_i}}{e^{\theta m_{i-1}}} & \text{using } m_i \leq m_{i-1}/3 \\
&\leq \frac{\epsilon e^{-t}}{|\Delta_i|^l} & \text{using } \eqref{eq:midecreasesfastenough}
\ea




We now complete the error analysis in Equation \eqref{equation:err}. Using the above equation
\ba
\mathcal{E}^{(l)}_i 
&\leq \sum_{p=0}^{m_{i-1} - l} \frac{\mathcal{E}^{(p+l)}_{i-1}}{p!} |\Delta_{i}|^{p} + \frac{\eps  e^{-t}}{|\Delta_i|^l}.
\ea

 We do the change of variable $\mathcal{F}^{(l)}_i = \mathcal{E}^{(l)}_i \cdot |\Delta_{i}|^l$. Using this notation this bound becomes
\ba
\mathcal{F}^{(l)}_i &\leq \sum_{p=0}^{m_{i-1} - l} \frac{\mathcal{F}^{(p+l)}_{i-1}}{p!} (\frac{|\Delta_{i}|}{|\Delta_{i-1}|})^{p + l} + \eps e^{-t}\\
&\leq \sum_{p=0}^{m_{i-1} - l} \frac{\mathcal{F}^{(p+l)}_{i-1}}{p!} + \eps e^{-t} & \text{using } |\Delta_i| \leq |\Delta_{i-1}|
\ea

Now define ${\cal F}_i = \max_{l} {\cal F}^{(l)}_i$. From the above, 
\ba
\mathcal{F}_i &\leq e \cdot \mathcal{F}_{i-1} + \eps e^{-t}.
\label{eq:error-version2}
\ea

The boundary condition is ${\cal F}_0 = 0$ since the derivatives are computed exactly at the first segment. Using \eqref{eq:error-version2}, by induction on $i \in [t]$ one can show that $\mathcal{F}_i \leq \frac{e^i - 1}{e-1} \cdot \eps e^{-t}$. We conclude that the output additive error is ${\cal E}_t^{(0)} = {\cal F}_t^{(0)} \cdot |\Delta_t|^0= {\cal F}_t^{(0)}  \leq {\cal F}_t \leq \epsilon \cdot e^{-t} \cdot e^t = \epsilon$.

\end{proof}

\subsection{Shallow Circuits for CAC}

In this section we establish that in fact Algorithm \ref{algorithm:cac} can be computed by shallow circuits.
To do that, we first note that the $k$ lowest  derivatives of $\log(g(z))$ can be computed efficiently from the lowest $k$ derivatives of $g(z)$:
\begin{lemma}\label{lem:brunodifaasubstitute}
Let $g(z)$ be an analytic function that is root-free in an open set $U$ containing $0$, and let $g(0) = 1$. Let $f(z) = \log(g(z))$. Then for each $k > 0$, there is an arithmetic circuit of fan-in 2 that receives as input the first $k$ derivatives of $g$ at $0$, $$g^{(0)}(0),\ldots,g^{(k)}(0),$$ and computes $f^{(k)}(0)$. Moreover, the circuit is of size $\poly(k)$, logspace-uniform, and computes a polynomial of degree $k$.
\end{lemma}

\begin{proof}
The Bruno di Fa\`a formula, which generalizes the chain rule to higher-order derivatives, states that given a composition of two functions $f(z) = h(g(z))$, the derivative $f^{(k)}(0)$ depends only on the first $k$ derivatives of $h$ at $z=g(0)=1$ and $g$ at $z=0$. In particular, we may define $$h(z) = \log(z), \quad \tilde{h}(z) = \sum_{i=1}^k \frac{(-1)^{i+1}}{i} \cdot (z-1)^i$$ and $$\tilde{g}(z) = 1+\sum_{i=1}^k \frac{g^{(i)}(0)}{i!} \cdot z^i,$$ and by Bruno di Fa\`a, $f(z) = h(g(z))$ will have the same $k$th derivative as $\tilde{f}(z) = \tilde{h}(\tilde{g}(z))$: $$f^{(k)}(0) = \tilde{f}^{(k)}(0).$$ 

Since $\tilde{f}(z)$ has a size-$O(k^2)$ logspace-uniform arithmetic circuit in $g^{(0)}(0),\ldots,g^{(k)}(0),z$, it follows by the derivative calculation lemma (Lemma \ref{lem:deriv}) that $\tilde{f}^{(k)}(0) = f^{(k)}(0)$ has a size-$\poly(k)$ logspace-uniform arithmetic circuit in $g^{(0)}(0),\ldots,g^{(k)}(0)$. Moreover, $\tilde{f}^{(k)}(0)$ is clearly of degree at most $k$ in $g^{(0)}(0),\ldots,g^{(k)}(0)$.
\end{proof}

We now use this lemma to establish that CAC can be computed by small circuits of low degree:
\begin{lemma}[Low-degree circuits for CAC]\label{lem:shallowCAC}

\noindent
%
Under the conditions of Lemma \ref{lem:caccorrectness}, Algorithm \ref{algorithm:cac} can be implemented by a logspace-uniform arithmetic circuit of size $\poly(m_0)$ that computes a polynomial of degree $O(m_0^2)$ in $$
g^{(0)}(0),\ldots,g^{(m_0)}(0),\Delta_1,\ldots,\Delta_t.
$$ 
\end{lemma}

\begin{proof}
Construct arithmetic circuits for $f^{(0)}(0),\ldots,f^{(m_0)}(0)$, the first $m_0$ derivatives of $f$ at $0$, using the procedure from Lemma \ref{lem:brunodifaasubstitute}. These circuits are logspace-uniform, and are of size at most $\poly(m_0)$ and degree at most $m_0$.

For each $0 \leq i \leq t$ and $0 \leq j \leq m_i$, construct a size-$(1+10im_0^3)$ arithmetic circuit computing $\hat{f}_i^{(j)}$, with degree at most $m_0-j+1$ in the variables $f^{(0)}(0),\ldots,f^{(m_0)}(0),\Delta_1,\ldots,\Delta_k$.  This construction is performed inductively on $i$. The base case $i=0$ is clear because $\hat{f}^{(j)}_0 = f^{(j)}(0)$.
For the inductive step, $\hat{f}_{i+1}^{(0)},\ldots,\hat{f}_{i+1}^{(m_{i+1})}$ can all be computed from $\hat{f}_i^{(0)},\ldots,\hat{f}_i^{(m_i)}$ using step \ref{it:iteration} of Algorithm \ref{algorithm:cac}, which can be implemented at an extra cost of $10m_i^3 \leq 10m_0^3$ gates. Moreover each $\hat{f}_{i+1}^{(j)}$ has degree $\max_{0 \leq p \leq m_i - j} \deg \hat{f}_i^{(p+j)} + p \leq (m_0-(p+j)+1+p) = m_0-j+1$ by the inductive hypothesis.

Let ${\cal C''}$ be the degree-$(m_0+1)$, size-$(1 + 10tm_0^3)$ circuit computing $\hat{f}_t^{(0)}$ in terms of 
$$
f^{(0)}(0), \ldots, f^{(m_0)}(0),
\quad
\mbox{and}
\quad 
\Delta_1, \ldots, \Delta_t.
$$ 
Compose ${\cal C'}$ with the degree-$m_0$, size-$\poly(m_0)$ circuits computing $f^{(0)}(0),\ldots,f^{(m_0)}(0)$ in terms of $g^{(0)}(0),\ldots,g^{(m_0)}(0)$ in order to obtain a degree-$O(m_0^2)$ 
size-$\poly(t,m_0) = \poly(m_0)$ circuit ${\cal C'}$ computing $\hat{f}_t^{(0)}$ in terms of $g^{(0)}(0),\ldots,g^{(m_0)}(0),\Delta_1,\ldots,\Delta_t$.
%

\end{proof}

Combining Lemma \ref{lem:shallowCAC} with Theorem \ref{thm:depth-reduction} (the depth-reduction theorem for arithmetic circuits) implies that the CAC algorithm can be computed by arithmetic circuits of $\poly(m_0)$ size, fan-in 2, and depth $O((\log m_0)^2)$. We will use this observation in the proof of the main theorem.



\section{Main results}\label{sec:alg}

\subsection{Theorem Statement}

Our main theorem is that for $O(\kappa)$-conditioned Hermitian or Hurwitz stable matrices one can compute a $1 + 1/\poly(n)$
approximation to the determinant using an arithmetic circuit of depth $\tilde {O}(\log (\kappa) \cdot \log(n))$.
Furthermore, this circuit can be implemented as a logspace-uniform Boolean circuit of polynomial size and $\tilde {O}(\log (\kappa) \cdot \log(n))$ depth as well as $\tilde {O}(\log (\kappa) \cdot \log(n))$ space:
\begin{theorem}[Approximation of the determinant of ${\cal H}_\delta$ and ${\cal S}_\delta$ matrices in near-$\NC^1$]
\ 

\noindent
For every $n$ and $\epsilon, \delta > 0$ there exists a logspace-uniform Boolean circuit 
of size $\poly(n)$ and depth $\tilde{O}((\log n) \cdot (\log(1/\delta) + \log \log(1/\epsilon)))$
such that for every input
$A\in {\cal H}_{\delta}$ it approximates $\Det(A)$ to multiplicative error $1 + \epsilon$.
In particular, for $\delta = 1/\poly\log (n)$, this circuit can be implemented in depth $\tilde{O}(\log(n))$.

The same result holds for ${\cal S}_\delta$ in place of ${\cal H}_{\delta}$.
\label{thm:main}
\end{theorem}
A direct corollary of Theorem \ref{thm:main} is the following:
\begin{corollary}[Approximation of the determinant of ${\cal H}_\delta$ and ${\cal S}_\delta$ matrices in near-$\Lspace$]
For every $n$ and $\epsilon, \delta > 0$, and $A \in \mathcal{H}_{\delta}$, there is a $\tilde{O}((\log n) \cdot (\log(1/\delta) + \log\log(1/\epsilon))$-depth algorithm that approximates $\Det(A)$ to multiplicative error $1 + \epsilon$.

The same result holds for ${\cal S}_\delta$ in place of ${\cal H}_{\delta}$.
\end{corollary}

\subsection{CAC interpolation points}

Recall the definition of the determinant interpolation polynomial $g_A(z) = \Det((1-z)I + zA)$. The proof of Theorem \ref{thm:main} will proceed by using Computational Anaytic Continuation (CAC) to approximate the value of $g_A(1) = \Det(A)$ from the low-order derivatives of $g_A(z)$ at $z = 0$.
\begin{lemma}[Interpolating segments for well-conditioned Hermitian matrices]\label{lem:CAC2}
\
\noindent
Let $\delta > 0$, let $A\in {\cal H}_{\delta}$, and let $g_A(z) = \Det( (1-z) I + zA)$.
Then there exist $t+1 = O(\log(1/\delta))$ CAC points $s_0,\hdots, s_t \in \C$ satisfying the conditions of Lemma \ref{lem:caccorrectness} with respect to $g_A$, with parameter $\theta > 0.4$.
\end{lemma}

\begin{proof}
Since $A$ is Hermitian, the roots of $g_A(z)$ for $A\in {\cal H}_\delta$ all lie on the real line. And by Theorem \ref{thm:roots}
we have that $g_A(z)$ is root-free in ${\cal D}(0,1/2) \cup {\cal D}(1,\frac{\delta}{1+\delta})$.
Consider CAC segments of 2 types:
\begin{enumerate}
\item
\textbf{Cross over:}
We cross from $0$ to $1+i/2$ above the real line using $6$ segments:
$$
s_0 = 0 
\to 
s_1 = 0.25 i 
\to 
s_2 = 0.5 i
\to
s_3 = 0.5 i + 0.25
\to
s_4 = 0.5 i + 0.5
$$
$$
\to
s_5 = 0.5 i + 0.75
\to 
s_6 = 0.5 i + 1 
$$
\item
\textbf{Decelerate:} We shuttle down from $s_6 = 1+i/2$ to $s_t = 1$ via a sequence of $O(\log 1/\delta)$ decreasing intervals. As we shuttle down, we reduce the interpolation disk radius on each step by a constant multiplicative factor. 
Let $t = \log_3(1/\delta) + O(1)$, $r_0 = 1/3$ and $b = 3$. We navigate 
$$
s_6 = 1 + i/2 \rightarrow \tilde s_7 = \tilde s_6 - i r_0 \rightarrow \tilde s_8 = \tilde s_7 - i r_0/b \rightarrow \ldots \rightarrow \tilde s_{t-1} = \tilde s_{t-2} - r_0 / b^{t-8}.
$$
More formally, for $6 \leq j \leq t-1$, we have $s_j = 1 + i/2 - \frac{i}{2}(1 - (1/3)^{j-6})$

At the end, move one more step from $\tilde s_{t-1}$ to $\tilde s_t = 1$. Note that in order to do this and still satisfy the CAC requirements we use $0 \leq \Im(s_{t-1}) \leq \delta/5$.
\end{enumerate}
We note that for each $j$ the polynomial is root-free in the disk ${\cal D}(s_j,(3/2) \cdot |s_{j+1}-s_j|)$. In particular, for $j \geq 6$ we have $|s_{j+1}-s_j| = (1/3)^{j-5}$, but the closest root to $s_i$ is on the real line, at least $(3/2)(1/3)^{j-5}$ distance away. For the segment from $s_{t-1} = 1 - ic \delta$ (for $0 \leq c \leq 1/5$) to $s_t = 1$, we use that $g_A$ is root-free in ${\cal D}(1,9\delta/10)$ Since $\log(3/2) > 0.4$, the bound on $\theta$ holds. Also, the segments are of non-increasing length and $g_A(0) = 1$, satisfying the other conditions of Lemma \ref{lem:caccorrectness}.
\end{proof}

\begin{lemma}
[Interpolating segments for well-conditioned Hurwitz stable matrices]\label{lem:CAC3}

\noindent
Let $\delta > 0$, let $A\in {\cal S}_{\delta}$, and let $g_{-A}(z) = \Det( (1-z) I - zA)$.
Then there exist $t+1 = O(\log(1/\delta))$ CAC points $s_0,\hdots, s_t \in \C$ satisfying the conditions of Lemma \ref{lem:caccorrectness} with respect to $g_{-A}$, with parameter $\theta > 0.4$.
\end{lemma}

\begin{proof}
The proof is very similar to the proof of Lemma \ref{lem:CAC2} we just presented. We first move to $z = 1/2$. This is doable because ${\cal P}_{1/2} = \{x : \Re(x) < 1/2\}$ and ${\cal D}(1/2,1/2)$ are root free by Theorem \ref{thm:roots}. Then, since ${\cal D}(1/2,1/2)$ and ${\cal D} (1,\delta/(1+\delta))$ are root free, we take a sequence of decelerating segments from $z = 0$ to $z=1/2$ with lengths shrinking by a constant factor at each step. 

Here is a way of doing this. Pick $t = \log_3(1/\delta) + O(1)$:
$$
s_0 = 0 \rightarrow s_1 = 1/6 \rightarrow s_2 = 1/3 \rightarrow s_3 = 1/2 \rightarrow s_4 = 1/2 + 1/6 \rightarrow s_5 = 1/2 + 1/3\rightarrow $$ $$s_6 = 1/2 + 1/3 + 1/3^2 \rightarrow \dots \rightarrow s_{t-1} = 1/2 + 1/3 + \ldots + 1/3^{t-5} \geq 1 - \delta/5 \rightarrow s_t = 1.
$$
More formally, for $5 \leq j \leq t-1$, we have $s_j = 1/2 + \frac{1}{2} (1 - (1/3)^{j-4})$. 

We note that for each $j$ the polynomial is root-free in the disk ${\cal D}(s_j,(3/2) \cdot |s_{j+1}-s_j|)$. In particular, for $j \geq 5$ we have $|s_{j+1}-s_j| = (1/3)^{j-5}$, but the closest root to $s_i$ lies outside $D(1/2,1/2)$, at least $(3/2)(1/3)^{j-5}$ distance away. For the segment from $s_{t-1} \geq 1 - \delta/5$ to $s_t = 1$ we use that $g_A$ is root-free in ${\cal D}(1,9\delta/10)$. Since $\log(3/2) > 0.4$, the bound on $\theta$ holds. Also, the segments are of non-increasing length and $g_{-A}(0) = 1$, satisfying the other conditions of Lemma \ref{lem:caccorrectness}. 
\end{proof}

\subsection{Proof of Theorem \ref{thm:main}}

Consider the following algorithm

\begin{mdframed}
\begin{algorithm}\label{alg:main}
\begin{enumerate}

\noindent
\item
\textbf{Input:}
$\delta > 0$, matrix $A \in {\cal H}_\delta$ or $A \in {\cal S}_\delta$.
\item
\textbf{Fixed parameters:}

\begin{enumerate}
\item $\theta = 0.4$ \hfill \% parameter in the CAC algorithm
\item $t = O(\log(1/\delta))$  \hfill \% number of CAC segments from $z = 0$ to $z = 1$
\item $k = \ceil{40\left(\log(n/\epsilon \theta)\right)\left(40t(\log t + \log \log(n/\epsilon \theta))\right)^{t}}$ \hfill \% number of derivatives CAC uses

\item
$r = k^{14}$ \hfill \% number of bits to which to round $A$
\item $M = k!$ \hfill \% size of the maximum constant used in the arithmetic circuits
\end{enumerate}

\item
\textbf{Main (for Hermitian $A \in \mathcal{H}_{\delta}$):}
\begin{enumerate}
\item If $k \geq n$, return the $\NC^2$-circuit exactly computing the determinant. Otherwise perform the following steps: 
\item Construct $C_{SB}$, the Samuelson-Berkowitz circuit (Theorem \ref{thm:det1}) computing $$g_A(z) = \Det((1-z) \cdot I + z \cdot A).$$ \hfill \
\item For each $0 \leq i \leq k$ construct $C_i = H_i(C_{SB})$, the arithmetic circuit computing the derivative $g_A^{(i)}(0)$ (using Lemma \ref{lem:deriv}). \hfill \
\item Construct the circuit $C_{CAC} = C_{CAC}(C_0, \ldots, C_k,\Delta_1,\ldots,\Delta_t)$ doing Computational Analytic Continuation from $z = 0$ to $z = 1$ as in Algorithm \ref{algorithm:cac} with steps $\Delta_1,\ldots,\Delta_t$, parameter $\theta$ and using the first $m_0 := k$ derivatives of $g_A$ at $z = 0$. (Lemma \ref{lem:shallowCAC})
\item Reduce the depth of the CAC circuit $C_{\mbox{low-depth}} = D(C_{CAC})$. (Theorem \ref{thm:depth-reduction})
\item Hard-code $\Delta_1, \hdots, \Delta_t$ to get the CAC points from Lemma \ref{lem:CAC2}.
\item Compute the Booleanization of the circuit $C_{\mbox{bool}} = B_{r,M}(C_{\mbox{low-depth}})$. (Lemma \ref{lem:Booleanization})
\item Return the Boolean circuit $C_{\mbox{out}} = \exp (C_{\mbox{bool}})$. \label{it:O}




\end{enumerate}

\item \textbf{Main (for Hurwitz stable $A \in \mathcal{S}_{\delta}$):} 

 The algorithm is essentially the same if $A \in {\cal S}_{\delta}$, but we use $g_{-A}(z)$ instead of $g_A(z)$, the interpolating segments are given by Lemma \ref{lem:CAC3} instead of Lemma \ref{lem:CAC2}, and we return $(-1)^{n}\cdot\exp(C_{\mbox{bool}})$ instead of $\exp(C_{\mbox{bool}})$, because $g_{-A}(1) = \Det(-A) = (-1)^{n} \Det(A)$.
\end{enumerate}
\vspace{5mm}
\end{algorithm}
\end{mdframed}

In order to prove correctness of Algorithm \ref{alg:main}, we first prove the following technical lemma:
\begin{lemma}[$C_{\mbox{low-depth}}$ has low depth] \label{lem:lowdepthmeanslowdepth}
If $k < n$, then $C_{\mbox{low-depth}}$ has size $\poly(n)$, degree $O(k^3)$, and depth $O(\log k)$. Each multiplication gate has fan-in 2 and each addition gate has fan-in at most $\poly(n)$. 
\end{lemma}
\begin{proof}
The Samuelson-Berkowitz circuit $C_{SB}$ constructed using Theorem \ref{thm:det1} is an arithmetic circuit of size $\poly(n)$.
By Lemma \ref{lem:deriv}, for all $0 \leq i \leq k$ the circuit $C_i$ is of size $\poly(n)$. Since $C_i$ computes the derivative of order $i\leq n$ w.r.t. the variable $z$ of $g_A(z)$ at $z = 0$,
$C_i$ has degree $O(k)$ in the entries of $A$, by Lemma \ref{lem:der-degree}.

Therefore, by Lemma \ref{lem:shallowCAC}, $C_{CAC}$ is of size $\poly(n)$ and has degree $O(k^3)$ in the entries of $A$ and in $\Delta_1,\ldots,\Delta_t$. It follows by Theorem \ref{thm:depth-reduction} (depth-reduction) that $C_{\mbox{low-depth}}$ is of size $\poly(nk) = \poly(n)$ and of depth $O(\log k)$, and that each multiplication gate has fan-in 2 and each addition gate has fan-in $\poly(n)$.
\end{proof}

\begin{lemma}[The circuit outputted by Algorithm \ref{alg:main} approximates the determinant of $A$] \label{lem:outputalgmaincorrect}

\noindent
Algorithm \ref{alg:main} computes a circuit $C_{\mbox{out}}$ that satisfies:
$$
C_{\mbox{out}}(A) = \Det(A) \cdot (1 + {\cal E}) , \ \ |{\cal E}| \leq \epsilon
$$
\end{lemma}

\begin{proof}
If $k \geq n$, then the algorithm computes the determinant exactly. Otherwise, by the error bound for CAC in Lemma \ref{lem:caccorrectness}, $C_{CAC}$ outputs an $\epsilon/4$ additive approximation to $\log(\Det(A))$ when the CAC segments from Lemma \ref{lem:CAC2} (respectively, Lemma \ref{lem:CAC3}) are hard-coded. Applying depth reduction (Theorem \ref{thm:depth-reduction}) does not change the output of $C_{CAC}$, and therefore $C_{\mbox{low-depth}}$ also computes an $\epsilon/4$ additive approximation.

We note that the constants used in the arithmetic circuit all have magnitude at most $k! = M$ (the largest constants are in the calculations of the derivatives by Lemma \ref{lem:deriv}), and the input variables have magnitude $\leq 1$. And by Lemma \ref{lem:lowdepthmeanslowdepth}, $C_{\mbox{low-depth}}$ is of size $\poly(n)$, degree $O(k^3)$, depth $O(\log k)$, has multiplication gates with fan-in 2, and addition gates with fan-in at most $n$. These are the preconditions to apply the Booleanization procedure (Lemma \ref{lem:Booleanization}). Since $r = k^{14} \geq k^{10} \cdot (\log n/\epsilon)^2 \log (k!) = \omega(k^9 \log(k) \log(n) + 1) \log(nk/\eps)(\log M))$, by the error bound in Lemma \ref{lem:Booleanization} we may conclude that the Booleanization procedure yields a Boolean circuit $C_{\mbox{bool}}$ that approximates $C_{CAC}$ up to additive $\epsilon/4$ error when the CAC points $s_1,\ldots,s_t$ are hard-coded. Hence overall $C_{\mbox{out}}$ is a $\exp(\epsilon/2)$ relative-error approximation of $\Det(A)$.
\end{proof}

\begin{lemma}\label{lem:bool2}
Algorithm \ref{alg:main} computes a $\poly(n)$-size Boolean circuit of depth $\tilde{O}(\log(n) \cdot (\log(1/\delta) + \log \log(1/\epsilon)))$.
\end{lemma}

\begin{proof}
If $k \geq n$, then the algorithm returns a size-$\poly(n)$, depth-$O((\log n)^2)$ circuit. In this case, $\log(k) \geq \log(n)$, so $t \log(t) + \log\log(1/\epsilon) = \Omega(\log(n))$, and since $t = \Theta(\log(1/\delta))$, the claim holds in this case.

In the case $k < n$, we also have $\delta > 1/n$. By Lemma \ref{lem:lowdepthmeanslowdepth} and Lemma \ref{lem:Booleanization}, we have that $C_{\mbox{bool}}$ is a circuit of size $\poly(nrk(\log k)(\log k!)) = \poly(n)$ and depth $O((\log k) \cdot \log(k(\log k)rn)) = O((\log n) \cdot (\log k)) = O((\log n) \cdot (t\log(t) + \log \log(n) + \log\log(\epsilon))) = \tilde{O}((\log n) \cdot (\log(1/\delta) + \log\log(1/\epsilon)))$. The final exponentiation operation is applied to a $\poly(\log(n/\epsilon)/\delta)$-bit number, and by the results of \cite{beame1986log,chiu2001division} it can be implemented by a logspace-uniform $\poly(n)$-size circuit of depth $O(\log(1/\delta) + \log\log(n/\epsilon))$ depth, which is negligible overhead.
\end{proof}

\begin{lemma} \label{lem:algmainlogspaceuniformity}
The circuit $C_{\mbox{out}}$ can be computed by Algorithm \ref{alg:main} in space $O(\log(n))$.
\end{lemma}

\begin{proof}
This follows from the fact that all of the operations involved can be done in logspace: computing $C_{BS}$ (Theorem \ref{thm:det1}), taking derivatives (Lemma \ref{lem:deriv}), CAC interpolation (Lemma \ref{lem:caccorrectness}), Booleanization (Lemma \ref{lem:Booleanization}), and, by \cite{beame1986log,chiu2001division}, taking the exponential.
\end{proof}

Theorem \ref{thm:main} follows from Lemmas \ref{lem:outputalgmaincorrect}, \ref{lem:bool2}, and \ref{lem:algmainlogspaceuniformity}.

\bibliographystyle{hyperabbrv}
\bibliography{refs}

\appendix

\section{Reduction from approximating $|\Det(A)|$ to approximate linear system solving} \label{app:psddetreductiontosolving}
In this section, we prove a near-$\NC^1$ reduction for approximating $|\Det(A)|$ based on approximately solving linear systems. Some implications of this reduction were mentioned in the introduction, but we go into more detail here. In contrast to the main result of this paper, this reduction does not recover the sign of $\Det(A)$.

The reduction from approximating $|\Det(A)|$ to approximately solving linear systems is based on the following proposition:
\begin{proposition}\label{prop:absolutevaluereduction}
Let $0 \prec A \preceq I$ be a positive definite $n \times n$ matrix. Let $A^{(1)},\ldots,A^{(n)}$ denote the principal submatrices of $A$: i.e., $A^{(i)}$ is the $i \times i$ submatrix consisting of the first $i$ rows and columns of $A$. Finally let $v_1,\ldots,v_n$ be approximations to $e_i^T(A^{(i)})^{-1}e_i$ such that for each $i$, $$|v_i - e_i^T (A^{(i+1)})^{-1} e_i| \leq \eps/(2n).$$

Then, for small enough $\eps > 0$, $\prod_{i=1}^n v_i$ is a $(1+\eps)$-multiplicative approximation to $\Det(A)$.
\end{proposition}
\begin{proof}
By Cauchy's interlacing theorem, $A^{(1)},\ldots,A^{(n)}$ have eigenvalues between $\lambda_n(A) > 0$ and $\lambda_1(A) \leq 1$. In particular, $A^{(1)},\ldots,A^{(n)}$ are non-singular, so we can write the telescoping product \begin{align*}\Det(A) &= \Det(A^{(n)}) \\ &= A_{1,1} \cdot \frac{\Det(A^{(n)})}{\Det(A^{(1)})} \\ &= A_{1,1} \cdot \prod_{i=1}^{n-1} \frac{\Det(A^{(i+1)})}{\Det(A^{(i)})} \\ &= \prod_{i=1}^{n} \frac{1}{e_i^T (A^{(i)})^{-1} e_i},\end{align*} where the last equality is by Cramer's rule.
For each $i$ the eigenvalues $A^{(i)}$ lie in of $(A^{(i)})^{-1}$ lie in $[1,\infty)$, by the Courant-Fischer min-max principle $e_i^T (A^{(i)})^{-1} e_i \geq 1$, so $v_i$ is a $(1-\eps/2n)$-multiplicative approximation of $e_i^T (A^{(i)})^{-1} e_i$. This concludes the proof.
\end{proof}

Suppose we are given a positive definite matrix $A$ and an algorithm that approximately solves systems of equations when the coefficient matrices are $A^{(1)},\ldots,A^{(n)}$, the principal submatrices of $A$. Then by Proposition \ref{prop:absolutevaluereduction} we can approximate $e_i^T (A^{(i)})^{-1} e_i$ in parallel for all $i \in [n]$, and multiply them together to approximate $|\Det(A)|$ with only near-$\NC^1$ overall overhead.

\subsection{Example applications of Proposition \ref{prop:absolutevaluereduction}}
We now review certain structured classes of well-conditioned matrices for which one can solve these systems of equations in low complexity:

\paragraph{Symmetric Diagonally Dominant (SDD)}

\cite{murtagh2017derandomization} gives a nearly-logspace solver for Symmetric Diagonally Dominant (SDD) matrices, a subclass of PSD matrices. If $A$ is SDD, then so are $A^{(1)},\ldots,A^{(n)}$. Therefore, if $A$ is a $\poly(n)$-conditioned symmetric diagonally dominant (SDD) matrix, Proposition \ref{prop:absolutevaluereduction} implies that $|\Det(A)|$ can be approximated in nearly-logspace. Since $A$ is PSD, in fact $|\Det(A)| = \Det(A)$ in this case.

\paragraph{Well-conditioned}
$A$ is a $\kappa$-conditioned matrix, then $B = A^{\dag} A$ is PSD and $\kappa^2$-conditioned. Moreover, by Cauchy's interlacing theorem, $B^{(1)},\ldots,B^{(n)}$ are also PSD and $\kappa^2$-conditioned. It suffices to show how to efficiently solve systems of linear equations with $\kappa^2$-conditioned PSD coefficient matrices.

In general, given a $\kappa^2$-conditioned PSD matrix $B$, then systems of equations $Bx = b$ can be approximately solved using gradient descent by outputting the approximation $\tilde{x} = \sum_{i=0}^{k-1} \alpha(1-\alpha A)^i b$, where $\alpha = 1/\kappa^2$ \cite{conjugategradientswithoutagonizingpain}. By repeated squaring, $\tilde{x}$ can be computed with a circuit of depth $\tilde{O}((\log n) \cdot (\log \kappa))$. So, using Proposition \ref{prop:absolutevaluereduction}, $\Det(B) = |\Det(A)|^2$ can be approximated in $\tilde{O}((\log n) \cdot (\log \kappa))$ depth, and therefore so can $|\Det(A)|$.

\paragraph{Hermitian stochastic}
If $A$ is a $\poly(n)$-conditioned Hermitian stochastic matrix, then $B = A^{\dag} A$ is PSD and stochastic. Moreover, we have $\|B^{(i)}\|_{\infty} \leq 1$ for all $i \in [n]$. (Here $\|B\|_{\infty} = \max_i \sum_{j=1}^n |B_{ij}|$.) For such PSD matrices $B$ with $\|B\|_{\infty} \leq 1$ the powers $B^k$ for $k = \poly(n)$ can be approximated in $\BPL$ \cite{doron2017approximating}. Therefore the gradient descent algorithm for solving $Bx = b$ can be run in $\BPL$. Hence by Proposition \ref{prop:absolutevaluereduction}, $\Det(B) = |\Det(A)|^2$ can be approximated in $\BPL$, and therefore so can $|\Det(A)|$.

\subsection{Quantum algorithm}
We also mention a quantum algorithm that outperforms the above classical algorithms for the case of Hermitian well-conditioned matrices:
\paragraph{Hermitian well-conditioned (quantum algorithm)}
In \cite{Ta13}, an algorithm for approximating the spectrum of $\kappa$-conditioned Hermitian matrices is given that runs in quantum space $O(\log(n) + \log(\kappa))$. In this case, the guarantee of the approximation is strong enough that simply multiplying the approximate eigenvalues gives a $(1+1/\poly(n))$ approximation to $\Det(A)$ in quantum space $O(\log(n) + \log(\kappa))$.


\section{Technical Estimates}



\begin{lemma} \label{lem:factorialdivbound}
For any $l,p \in \mathbb{Z}_{>0}$, 
$$\frac{(l+p-1)!}{p!} \leq e \cdot (p/e)^l (\frac{p+l}{p})^{p+l}$$
\end{lemma}
\begin{proof}
\ba
\frac{(l+p-1)!}{p!} &\leq \frac{1}{p!} \cdot \frac{(l+p-1)^{l+p}}{e^{l+p-2}} & \text{using } n! \leq \frac{n^{n+1}}{e^{n-1}} \\
&\leq \frac{e^{p-1}}{p^p} \cdot \frac{(l+p-1)^{l+p}}{e^{l+p-2}} & \text{using } n! \geq \frac{n^n}{e^{n-1}} \\
&\leq e \cdot (p/e)^l (\frac{p+l}{p})^{p+l}
\ea
\end{proof}

\begin{lemma}\label{lem:int1}
For all $\beta > 1$ and $m > l/\log \beta$ 
\be
\sum_{k=m}^{\infty} \beta^{-k} \cdot k^l
\leq m^l \beta^{-m} \cdot \frac{1}{1-\beta^{-1}\cdot e^{l/m}}
\ee
\end{lemma}
\begin{proof}
Let $a_k = \beta^{-k} \cdot k^l$. For $k \geq m$ we have $\frac{a_{k+1}}{a_k} = \frac 1 \beta (1 + 1/k)^l \leq \frac 1 \beta e^{l/m} =: \alpha$. Therefore,
$$
\sum_{k=m}^{\infty} \beta^{-k} \cdot k^l \leq a_m \cdot \sum_{k=0}^\infty \alpha^k \leq m^l \beta^{-m} \cdot \frac{1}{1-\beta^{-1}\cdot e^{l/m}}.
$$
\end{proof}



\section{Booleanization details} \label{app:Booleanizationdetails}
Recall that given an arithmetic circuit ${\cal C}$ over $\bbC$ on variables $x_1,\ldots,x_N$, the Booleanization $B_{r,M}({\cal C})$ is the Boolean circuit constructed by assuming that all inputs $x_1,\ldots,x_N$ have magnitude $|x_i| \leq M$, and rounding them to $r$ bits of precision using the operation:\footnote{The floor function is applied to the real and imaginary parts separately.} $$R_r(z) := \floor{2^r \cdot z} / 2^r.$$

In this section, we will prove Lemma \ref{lem:Booleanization}, restated below:
\begin{lemma}
Let $\eps > 0$, and let ${\cal C}$ be a circuit over $\bbC$ of depth $h$, computing a polynomial $g(x_1,\ldots,x_N)$ of degree $d$. Suppose that each multiplication gate of ${\cal C}$ is of fan-in 2, and each addition gate is of fan-in at most $m$. For technical reasons, suppose that all input gates of ${\cal C}$ are labelled by a variable in $\{x_1,\ldots,x_N\}$ (i.e., there are no input gates labelled by a constant in $\bbC$).

If $r > (2hd^2 \ceil{\log(m)} + 1) \log_2(4NMd/\eps)$, then $B_{r,M}({\cal C})$ is a logspace-uniform Boolean circuit of size $\poly(|\mathcal{C}|dhr(\log m)\log(M))$ and depth-$O(h \cdot \log(dhrmM))$. Moreover, $B_{r,M}({\cal C})$ computes a function $\tilde{g}(x_1,\ldots,x_N)$ such that for all $a_1,\ldots,a_N \in \bbC$ with $\max_i |x_i| \leq M$, $$|\tilde{g}(a_1,\ldots,a_n) - g(a_1,\ldots,g_n)| < \eps.$$
\end{lemma}
\begin{proof}
This follows from Lemmas \ref{lem:boundingrounding} and \ref{lem:depthcomplexitybound}, which are proved in Subsections \ref{subsec:Booleanizationdepth} and \ref{subsec:roundingaccuracy}.
\end{proof}

In Subsection \ref{subsec:roundingaccuracy}, we ensure that the function computed by $B_{r,M}({\cal C})$ is a good approximation of the polynomial computed by ${\cal C}$. And in Subsection \ref{subsec:Booleanizationdepth} we bound the depth of $B_{r,M}({\cal C})$. This requires bounding the number of bits required to represent the values in the intermediate computation. In both Subsections \ref{subsec:roundingaccuracy} and \ref{subsec:Booleanizationdepth}, we will use the following lemma:

\begin{lemma}[Bound on circuit value] Let ${\cal C}$ be an arithmetic circuit over $\bbC$. Suppose that each multiplication gate of ${\cal C}$ is of fan-in 2, and each addition gate is of fan-in at most $m$. For technical reasons, suppose that all input gates of ${\cal C}$ are all labelled with variables in $\{x_1,\ldots,x_N\}$ (i.e., there are no constants from $\bbC$ in the input gates).

For each node $v$ at height $h(v)$ in ${\cal C}$, let $p_v(x_1,\ldots,x_N)$ denote the polynomial of degree $d(v)$ computed at $v$. Then if $\max_i |x_i| \leq M$, 
\be
|p_v(x_1,\ldots,x_N)| \leq f(d(v), h(v)) := (2M)^{d(v) h(v) \ceil{\log_2(m)} + 1}.
\ee
\label{lem:precisionbound}
\end{lemma}
\begin{proof}
First we note that we may assume that $m = 2$ without loss of generality, because each fan-in-$m$ addition gate can be replaced by a depth-$\ceil{\log m}$ tree of fan-in-2 addition gates, increasing the depth of the circuit by at most a factor of $\ceil{\log m}$.
The proof is by induction on $h(v)$, the height of $v$. For the base case, $v$ is an input gate and $h(v) = 0$, $d(v) = 1$, since $p_v = x_i$ for some $|x_i|\leq M$,
\be
|p_v(x_1,\ldots,x_N)| \leq |x_i| \leq (2M)^{d(v) h(v) + 1}.
\ee
For the inductive step, if $v$ is not an input gate, let $w_1$ and $w_2$ be its children at heights $h(w_1), h(w_2) \leq h(v)-1$. If $v$ is a multiplication gate, then
\begin{align*}
|p_v(x)| &= |p_{w_1}(x)| \cdot |p_{w_2}(x)|\\
&\leq (2M)^{d(w_1) h(w_1) + 1} \cdot (2M)^{d(w_2) h(w_2) + 1}\\
&\leq (2M)^{d(w_1) (h(v)-1) + 1 + d(w_2) (h(v)-1) + 1}\\
&= (2M)^{d(v) h(v) + 2-d(v)} \tag{Using $d(w_1) + d(w_2) = d(v)$}\\
&\leq (2M)^{d(v) h(v) + 1}.
\end{align*}

And if $v$ is an addition gate then
\begin{align*}
|p_v(x)| &= |p_{w_1}(x)| + |p_{w_2}(x)|\\
&\leq (2M)^{d(w_1) h(w_1) + 1} + (2M)^{d(w_2) h(w_2) + 1}\\
&\leq 2 \cdot (2M)^{d(v) (h(v)-1) + 1} \tag{By $h(w_1),h(w_2) \leq h(v)$ and $d(w_1),d(w_2) \leq d(v)$}\\
&\leq  (2M)^{d(v) h(v)+1}.
\end{align*}
\end{proof}

\subsection{Bounding the error from rounding} \label{subsec:roundingaccuracy}
A corollary to this lemma is that we can round the input values to a low number of bits of precision, and incur only a small additive error.
\begin{lemma}[Bound on rounding error] \label{lem:boundingrounding}
Let ${\cal C}$ be an arithmetic circuit over $\bbC$ of depth $h$ and degree $d$ such that all input gates are labelled by input variables $\{x_1,\ldots,x_N\}$ and not constants in $\bbC$. Suppose that each multiplication gate of ${\cal C}$ is of fan-in 2, and each addition gate is of fan-in at most $m$.

Let $g(x_1,\ldots,x_N) \in \bbC[x_1,\ldots,x_N]$ be the polynomial computed by ${\cal C}$. Let $M > 0$, $0 \leq \eps < 1$, $a_1,\ldots,a_n \in \bbC$ and $b_1,\ldots,b_n \in \bbC$ be such that \begin{equation*}\max_i |b_i - a_i| \leq \eps, \quad \mbox{ and } \quad \max_i |a_i| \leq M, ,\max_i |b_i| \leq M.\end{equation*} Then $$|g(a_1,\ldots,a_n) - g(b_1,\ldots,b_n)| \leq Nd\eps(2M)^{2hd^2\ceil{log m}+1}$$.
\end{lemma}
\begin{proof} Assume $m = 2$ without loss of generality, because each fan-in-$m$ addition gate can be replaced by a tree of addition gates of depth $\ceil{\log m}$.
For each $i \in [N]$ and $0 \leq j \leq d$, consider the polynomial $$[x_i^j]g(x_1,\ldots,x_N) \in \bbC[x_1,\ldots,x_{i-1},x_{i+1},\ldots,x_N].$$ By Lemma \ref{lem:deriv}, there is a depth-$(2hd)$ arithmetic circuit ${\cal C}_{i,j}$ computing $[x_i^j]g$. Moreover, the construction in Lemma \ref{lem:deriv} does not add any field elements from $\bbC$ to the input gates. Therefore, by Lemma \ref{lem:precisionbound} we have the following inequality:
$$|[x_i^j]g(a_1,\ldots,a_{i-1},a_{i+1},\ldots,a_N)| \leq (2M)^{2hd^2+1}$$
Therefore, for all $i \in [N]$, defining
$\Delta_i := g(a_1,\ldots,a_i,b_{i+1},\ldots,b_N) - g(a_1,\ldots,a_{i-1},b_i,\ldots,b_N)$, we have \begin{align*}|\Delta_i| \leq \sum_{j=0}^d |b_i - a_i|^j \cdot |[x_i^j]g(a_1,\ldots,a_{i-1},b_{i+1},\ldots,b_N)| \leq d \eps (2M)^{2hd^2 + 1}.\end{align*}
So \begin{align*}|g(a_1,\ldots,a_n) - g(b_1,\ldots,b_n)| = |\sum_{i \in [N]} \Delta_i| \leq \sum_{i \in [N]} |\Delta_i| \leq Nd\eps(2M)^{2hd^2 + 1}.\end{align*}
\end{proof}

A corollary of Lemma \ref{lem:boundingrounding} is the following:
\begin{corollary}\label{cor:boundingroundingimplication}
Let $\eps,M > 0$, and let $g(x_1,\ldots,x_N)$ be as in the statement of Lemma \ref{lem:boundingrounding}. If $r > (10hd^3\ceil{\log m}+1)\log_2(4NdM/\eps)$, then for $x_1,\ldots,x_N \in \bbC$ such that $\max_i |x_i| \leq M$, $$|g(R_r(x_1),\ldots,R_r(x_N)) - g(x_1,\ldots,x_N)| \leq 2^{-r+1} \cdot Nd(2M)^{10hd^3\ceil{\log m} + 1} \leq \eps.$$
\end{corollary}
For example, if $m = 2$ and the degree of $g$ is $\poly\log(N)$ and the height is $\poly\log(N)$, then Corollary \ref{cor:boundingroundingimplication} implies that we can round the inputs of the circuit down to $\poly\log(N) \cdot \log(M)$ bits of precision and still obtain a $1/\poly(N)$ overall approximation to the true value.

\subsection{Bounding the depth of the Booleanization} \label{subsec:Booleanizationdepth}
In this subsection, we show that replacing each arithmetic operation in an arithmetic circuit with its fixed-precision Boolean analogue does not increase the depth of the circuit significantly. The parallel Boolean complexity of the basic arithmetic operations is a folklore result:

\begin{lemma} [Boolean complexity of addition, multiplication and exponentiation]
Complex multiplication and iterated addition are contained in $\NC^1$.
That is, 
\begin{enumerate}
\item If $a_1,\ldots,a_t \in \mathbb{C}$ are $t$ complex numbers in $t$ bits then $\sum_{i=1}^t a_i$ can be computed using a bounded fan-in Boolean logspace-uniform circuit of size $\poly (t)$ and depth $O(\log t)$.
\item If $a,b \in \mathbb{C}$ are complex numbers in $t$ bits, then $a \times b$ can be computed using a bounded fan-in Boolean logspace-uniform circuit of size $\poly (t)$ and depth $O(\log t)$.
\end{enumerate}
\label{lem:simple-operations-cost}
\end{lemma}
Recall that $B_{r,M}({\cal C})$ is defined to be the Boolean circuit formed by rounding the input variables to $r$ bits of precision, and replacing each arithmetic operation with a Boolean operation, assuming that at the input gates we have $\max_i |x_i| \leq M$.
\begin{lemma} \label{lem:depthcomplexitybound}
Let $r > 0$, and let ${\cal C}$ be an arithmetic circuit over $\bbC$ of depth $h$ and degree $d$ such that the input gates are labelled by variables $\{x_1,\ldots,x_N\}$ and not constants in $\bbC$. Suppose that each multiplication gate of ${\cal C}$ is of fan-in 2, and each addition gate is of fan-in at most $m$.

Then $B_{r,M}({\cal C})$ has depth $O(h\log (dhrmM))$.
\end{lemma}
\begin{proof}
We will prove this lemma by bounding the bit complexity of the values at the intermediate gates $v$ of the arithmetic circuit by $B = O(d^3 hr (\log m)(\log M))$. This will suffice, because when the number of bits at each gate is bounded by $B$ then replacing each addition or multiplication operation with the Boolean implementation adds only $\log(mB)$ depth to the circuit by Lemma \ref{lem:simple-operations-cost}. Thus the total depth is $O(h \log(mB)) = O(h \log (dhrmM))$, as claimed.

For each node $v$ in ${\cal C}$, let $p_v(x_1,\ldots,x_N)$ be the polynomial computed at $v$. Write $$p_v(x_1,\ldots,x_N) = \sum_{i=0}^d p_{v,i}(x_1,\ldots,x_N),$$ where for each $0 \leq i \leq d$, $p_{v,i}$ is a homogeneous polynomial of degree $i$ (i.e., each of its monomials is of degree $i$). By arithmetic circuit homogenization (cf. Lemma \ref{lem:deriv}), for each $0 \leq i \leq d$ there is an arithmetic circuit of depth $h(d+1)$ computing $p_{i,v}$ such that all input gates are labelled by variables in $\{x_1,\ldots,x_N\}$. 

We apply Lemma \ref{lem:precisionbound} to $p_{v,i}(x_1 \cdot 2^r,\ldots,x_N \cdot 2^r) = p_{v,i}(x_1,\ldots,x_N) \cdot 2^{ir}$, and conclude that $|p_{v,i}(x_1 \cdot 2^r,\ldots,x_N \cdot 2^r)| \leq (2M)^{(r+1)(2d^2\ceil{\log_2(m)} h+1)}$.
Since $2^{r d} p_{v,i}(x_1,\ldots,x_N)$ is a Gaussian integer, we conclude that only $O((r+1)(d^3 h+1)(\log M))$ bits of precision are required to represent $p_{v,i}(x_1 \cdot 2^r,\ldots,x_N \cdot 2^r)$. Hence only $O(d(v) \cdot (r+1)(d^2 h+1)(\log M))$ bits are required to represent $p_v(x_1,\dots,x_N)$. So overall, all intermediate values in $B_{r,M}({\cal B})$ can indeed be represented with $B = O(d^3 hr(\log m)(\log M))$ bits of precision.
\end{proof}

\end{document}